\documentclass[onecolumn,12pt]{IEEEtran}
\usepackage{graphicx,epic,eepic,epsfig,amsmath,amsfonts,amsthm,latexsym,amssymb,color,enumerate,bbm}
\allowdisplaybreaks[3]

\usepackage[colorlinks,
            linkcolor=blue,
            anchorcolor=blue,
            citecolor=red,
            ]{hyperref}

\newtheorem{theorem}{Theorem}

\newtheorem{lemma}[theorem]{Lemma}
\newtheorem{proposition}[theorem]{Proposition}

\theoremstyle{definition}
\newtheorem{definition}[theorem]{Definition}
\theoremstyle{remark}
\newtheorem{remark}{Remark}

\def\supp{\mathop{\rm supp}}

\def\ox{{\otimes}}

\newcommand{\tr}{\operatorname{Tr}}
\newcommand{\mc}{\mathcal}

\newcommand{\beq}{\begin{equation}}
\newcommand{\eeq}{\end{equation}}

 \newenvironment{proofof}[1]{\vspace*{5mm} \par \noindent
{\it Proof of #1:\hspace{2mm}}}{\qed
}

\begin{document}
\title{Strong Converse Exponent of Quantum Dichotomies}

\author{
Mario~Berta,
Yongsheng Yao  \thanks{Mario Berta is with the Institute for Quantum Information, RWTH Aachen University, Aachen 52074, Germany. Yongsheng Yao is with the Institute for Quantum Information, RWTH Aachen University, Aachen 52074, Germany. (yongsh.yao@gmail.com). }}
\date{}

\maketitle

\begin{abstract}
    The quantum dichotomies problem asks at what rate one pair of quantum states can be approximately mapped into another pair of quantum states. In the many copy limit and for vanishing error, the optimal rate is known to be given by the ratio of the respective quantum relative distances. Here, we study the large-deviation behavior of quantum dichotomies and determine the exact strong converse exponent based on the purified distance. This is the first time to establish the exact high-error large-deviation analysis for this task in fully
quantum setting.
\end{abstract}

\begin{IEEEkeywords}
strong converse exponent, quantum dichotomies, sandwiched R\'enyi divergence, purified distance, trace distance
\end{IEEEkeywords}

\section{Introduction}

 Consider two quantum dichotomies $(\rho_1,\sigma_1)$ and $(\rho_2,\sigma_2)$, quantum dichotomies transformation is a procedure of using a quantum channel to transform $(\rho_1,\sigma_1)$ into as much $(\rho_2,\sigma_2)$ as possible, in which we require that the transformation $\sigma_1 \rightarrow \sigma_2$ is exact and the transformation $\rho_1 \rightarrow \rho_2$ is allowed to be approximate~\cite{FDT2019information, WangWilde2019resource}. It is an important quantum information processing task which have wide variety of applications in quantum resource theory~\cite{FDT2019information, WangWilde2019resource,LPCRTK2024quantum, KCT2019avoid, KCT2018Beyond, KCT2019Moderate, KumagaiHayashi2013second, Wilde2022distinguishability}. The commonly used distances to measure the closeness between the final state and the target state $\rho_2$ are the trace distance and the purified distance. In the asymptotic setting when unlimited $(\rho_1,\sigma_1)$ are available, no matter which distance we choose, the optimal rate at which the asymptotically perfect transformation can be achieved equals to the ratio between the quantum relative entropy $D(\rho_1\|\sigma_1)$ and $D(\rho_2 \| \sigma_2)$~\cite{FDT2019information, WangWilde2019resource}.

Although the first-order asymptotic analysis of quantum dichotomies transformation has been well understood, the finer asymptotic results, including the second-order expansion, moderate-deviation analysis and large-deviation type exponential analysis, are more complicated and they depend on the choice of the distance. The reference~\cite{LPCRTK2024quantum} established the exact second-order asymptotic, moderate-deviation expansion and large-deviation type exponential behavior based on the trace distance under the assumption $\rho_2$ commutes with $\sigma_2$. The reference~\cite{KCT2018Beyond} and~\cite{KCT2019Moderate} established the second-order asymptotic and the moderate-deviation analysis based on the purified distance, respectively,  when $(\rho_1,\sigma_1)$ and $(\rho_2,\sigma_2)$ are both classical dichotomies. To the best of our knowledge, the current research on the finer asymptotic analysis of quantum dichotomies transformation are only restricted in classical or semi-classical setting. The related results are still unknown in fully quantum setting.

In this paper, we study the high-error large-deviation exponential behavior for quantum dichotomies transformation, also known as the strong converse exponent for this task, which is the best rate of exponential convergence of the performance of this task towards the worst. We prove that for any quantum dichotomies $(\rho_1,\sigma_1)$ and $(\rho_2,\sigma_2)$, when the transformation rate $r$ is larger than the first-order asymptotic  $D(\rho_1\|\sigma_1)\cdot D(\rho_2\|\sigma_2)^{-1}$, the strong converse exponent based on the purified distance is given by
\begin{equation}
\sup_{\frac{1}{2} \leq \alpha \leq 1} \frac{1-\alpha}{\alpha}\Big\{rD_\alpha^*(\rho_2\|\sigma_2)-D_{\alpha/(2\alpha-1)}^*(\rho_1\|\sigma_1) \Big\},
\end{equation}
where $D_\alpha^*(\rho\|\sigma)$ is the sandwiched R\'enyi divergence. Our work is the first time to establish the exact strong converse exponent for this task in fully quantum setting, filling the gap in previous work.


In addition, if $\rho_2$ is a pure state, we also show that the strong converse exponent based on the trace distance is 
\begin{equation}
\sup_{\beta \geq 1} \frac{\beta-1}{\beta}\Big\{rD_{1/\beta}(\rho_2 \| \sigma_2)-D_{\beta}^*(\rho_1\|\sigma_1)  \Big\},
\end{equation}
where $D_\alpha(\rho\|\sigma)$ is the Petz R\'enyi divergence. We conjecture that this formula still holds for general quantum states $\rho_2$. The Petz R\'enyi divergence was used to characterize the reliability functions~(the best rate of
exponential convergence of the error of a task towards the perfect) for 
quantum information processing tasks in previous research~\cite{AKCMBMA2007discriminating, NussbaumSzkola2009chernoff, Li2016discriminating, Nagaoka2006converse, Hayashi2007error, ANSV2008asymptotic, Renes2022achievable, CHDH2020non, HayashiTomamichel2016correlation}. If the conjecture is correct, this will provide new type of operational meaning for the Petz R\'enyi divergence.

In the derivation of the achievability part of the strong converse exponent based on the purified distance, we use a technique developed by Mosonyi and Ogawa~\cite{MosonyiOgawa2017strong}, to derive a weaker bound characterized by the log-Euclidean R\'enyi divergence and then improve it to the final bound, one of the authors also use the similar technique to establish the strong converse exponents for privacy amplification and quantum information decoupling~\cite{LiYao2024operational}. The proof of the optimality part comes from an application of the operator H\"{o}lder inequality and this part can also be deduced from the work of Wang and Wilde~\cite{WangWilde2019resource}. For the strong converse exponent based on the trace distance, we mainly make use of an improved Fuchs-van de Graaf to relate it to the strong converse exponent based on the purified distance.

The remainder of this paper is organized as follows. In Section~\ref{sec:pre}, we introduce some preliminaries in quantum information theory. In Section~\ref{sec:main}, we present the main problem and the main result. In Section~\ref{sec:achi}, we give the proof of the achievability part of the strong converse exponent for quantum dichotomies transformation.  In Section~\ref{sec:conclu}, we conclude this paper with some discussion.

\section{Preliminaries}
\label{sec:pre}

\subsection{Notation}

Let $\mathcal{H}$ be a finite-dimensional Hilbert space, we denote the set of linear operators on $\mathcal{H}$ as $\mathcal{L}(\mathcal{H})$. Let $\mathcal{P}(\mathcal{H})$ and $\mathcal{S}(\mathcal{H})$ represent the set of 
positive semi-definite operators and the set of quantum states, respectively. The support of an operator $A \in \mathcal{L}(\mathcal{H})$ is denoted as $\supp(A)$. For $\rho \in \mathcal{S}(\mathcal{H})$, we use the notation
$\mathcal{S}_\rho$ for the set of quantum states whose supports are contained in $\supp(\rho)$. We denote as $|\mathcal{H}|$ the dimension of $\mathcal{H}$.

The trace distance and purified distance are two most commonly used distances to measure the closeness of two quantum states. For $\rho, \sigma \in \mathcal{S}(\mathcal{H})$, the trace distance $d(\rho,\sigma)$  and the 
purified distance $P(\rho,\sigma)$ are given by 
\begin{align}
\begin{split}
&d(\rho,\sigma):=\frac{\|\rho-\sigma\|_1}{2}, \\
&P(\rho,\sigma):=\sqrt{1-F^2(\rho,\sigma)},
\end{split}
\end{align}
where $\|\cdot\|_1$ is the Schatten-1 norm and $F(\rho,\sigma):=\|\sqrt{\rho}\sqrt{\sigma}\|_1$ is the fidelity function.

A quantum channel $\mathcal{N}$ is a completely positive and trace-preserving map. Let $A$ be a self-adjoint operator with spectral projections $P_1, \ldots, P_r$. The pinching channel $\mathcal{E}_A$ associated with $A$ is defined as
\[
\mathcal{E}_A : X\mapsto\sum_{i=1}^r P_i X P_i.
\]
The pinching inequality~\cite{Hayashi2002optimal} tells that for any $\sigma \in \mathcal{P}(\mathcal{H})$, we have
\begin{equation}
\sigma \leq v(A) \mathcal{E}_A(\sigma),
\end{equation}
where $v(A)$ is the number of different eigenvalues of $A$.


\subsection{Quantum R\'enyi divergences}

The classical R\'enyi divergence~\cite{Renyi1961measures} is an information quantity which play important roles in many fields. Because of the non-commutativity of quantum states, there are infinite in-equivalent quantum generalizations of classical R\'enyi divergence. In this section, we introduce the sandwiched R\'enyi divergence~\cite{MDSFT2013on, WWY2014strong}, the Petz R\'enyi divergence~\cite{Petz1986quasi} and the log-Euclidean R\'enyi divergence~\cite{MosonyiOgawa2017strong, HiaiPetz1993the} which have been found wide applications so far~\cite{AKCMBMA2007discriminating, NussbaumSzkola2009chernoff, Li2016discriminating, Nagaoka2006converse, Hayashi2007error, ANSV2008asymptotic, Renes2022achievable, CHDH2020non, HayashiTomamichel2016correlation, MosonyiOgawa2015quantum, GuptaWilde2015multiplicativity, CMW2016strong, LiYao2024strong, LYH2023tight, LiYao2024reliability, Hayashi2015precise, LiYao2021reliable}.

\begin{definition}
\label{definition:sand}
Let $\alpha\in(0,+\infty)\setminus\{1\}$, $\rho\in\mc{S}(\mc{H})$ and $\sigma\in\mc{P}(\mc{H})$.
When $\alpha >1$ and $\supp(\rho)\subseteq\supp(\sigma)$ or $\alpha\in (0,1)$ and $\supp(\rho)\not\perp\supp(\sigma)$, the sandwiched R{\'e}nyi divergence and the Petz R\'enyi divergence of order $\alpha $
are respectively defined as
\begin{align}
\begin{split}
&D_{\alpha}^*(\rho \| \sigma):=\frac{1}{\alpha-1} \log Q_{\alpha}^*(\rho \| \sigma),
\quad\text{with}\ \
Q_{\alpha}^*(\rho \| \sigma)=\tr {({\sigma}^{\frac{1-\alpha}{2\alpha}} \rho {\sigma}^{\frac{1-\alpha}{2\alpha}})}^\alpha, \\
&D_{\alpha}(\rho \| \sigma):=\frac{1}{\alpha-1} \log Q_{\alpha}(\rho \| \sigma),
\quad\text{with}\ \
Q_{\alpha}(\rho \| \sigma)=\tr \rho^\alpha \sigma^{1-\alpha};
\end{split}
\end{align}
otherwise, we set $D_{\alpha}^*(\rho \| \sigma)=D(\rho\|\sigma)=+\infty$. When $\alpha > 1$ and $\supp(\rho)\subseteq\supp(\sigma)$  or  $\alpha \in (0,1)$ and $\supp(\rho)\cap\supp(\sigma)\neq\{0\}$, the log-Euclidean R{\'e}nyi divergence of order
$\alpha $ is defined as
\beq
D_{\alpha}^{\flat}(\rho \| \sigma):=\frac{1}{\alpha-1}\log Q_{\alpha}^{\flat}(\rho \| \sigma),
\quad\text{with}\ \
Q_{\alpha}^{\flat}(\rho \| \sigma)=\tr 2^{\alpha \log \rho +(1-\alpha) \log \sigma};
\eeq
otherwise, we set $D_{\alpha}^{\flat}(\rho \| \sigma)=+\infty$. 
\end{definition}
When $\alpha \rightarrow 1$, all above quantum R\'enyi divergences converge to the Umegaki relative entropy~\cite{Umegaki1954conditional}
\beq
D(\rho\|\sigma):= \begin{cases}
\tr\rho(\log\rho-\log\sigma) & \text{ if }\supp(\rho)\subseteq\supp(\sigma), \\
+\infty                        & \text{ otherwise.}
                  \end{cases}
\eeq

In the following proposition, we collect some important properties of these quantum R\'enyi divergences.

\begin{proposition}
\label{prop:mainpro}
Let $\rho \in \mc{S}(\mc{H})$ and $\sigma \in \mc{P}(\mc{H})$. The sandwiched R{\'e}nyi
divergence, the Petz R\'enyi divergence and the log-Euclidean R{\'e}nyi divergence satisfy the following properties.
\begin{enumerate}[(i)]
  \item Monotonicity in R{\'e}nyi parameter~\cite{MosonyiOgawa2017strong, MDSFT2013on,Beigi2013sandwiched}: if $0\leq \alpha \leq \beta$, then
      $D_{\alpha}^{(t)}(\rho \| \sigma) \leq  D^{(t)}_{\beta}(\rho \| \sigma)$,
      for $(t)=*$, $(t)=\{\}$ and $(t)=\flat$;
  \item Monotonicity in $\sigma$~\cite{MosonyiOgawa2017strong, MDSFT2013on}: if $\sigma' \geq \sigma$, then $D_{\alpha}^{(t)}(\rho \| \sigma') \leq D_{\alpha}^{(t)}(\rho \| \sigma)$,
      for $(t)=*$, $\alpha \in [\frac{1}{2},+\infty)$, for $(t)=\{\}$, $\alpha \in [0,2]$ and for $(t)=\flat$, $\alpha \in [0,+\infty)$;
  \item Variational representation~\cite{MosonyiOgawa2017strong}: the log-Euclidean
      R{\'e}nyi divergence has the following variational representation
      \beq
      D_{\alpha}^{\flat}(\rho \| \sigma)= \begin{cases}
         \min\limits_{\tau \in \mc{S}(\mc{H})} \big\{D(\tau \| \sigma)
         -\frac{\alpha}{\alpha-1}D(\tau \| \rho)\big\}, & \alpha \in (0,1), \\
         \max\limits_{\tau \in \mc{S}_\rho(\mc{H})} \big\{D(\tau \| \sigma)
         -\frac{\alpha}{\alpha-1}D(\tau \| \rho)\big\}, & \alpha \in (1,+\infty);
      \end{cases}
      \eeq
  \item Data processing inequality~\cite{MosonyiOgawa2017strong, MDSFT2013on, WWY2014strong, Beigi2013sandwiched, FrankLieb2013monotonicity}: letting $\mc{N}$ be a CPTP map from $\mc{L}(\mc{H})$ to $\mc{L}(\mc{H}')$, we have
      \beq
      D_{\alpha}^{(\rm{t})}(\mc{N}(\rho) \| \mc{N}(\sigma)) \leq D_{\alpha}^{(\rm{t})}(\rho \| \sigma),
      \eeq
      for $(t)=\flat$, $\alpha \in [0,1]$, for $(t)=\{\}$, $\alpha \in [0,2]$ and for $(t)=*$, $\alpha \in [\frac{1}{2},+\infty)$;
 \item Approximation by pinching~\cite{HayashiTomamichel2016correlation, MosonyiOgawa2015quantum}: for any $\alpha \geq 0$, we have
      \beq
      D_{\alpha}^*(\mc{E}_\sigma(\rho) \| \sigma) \leq D_{\alpha}^*(\rho \| \sigma) \leq D_{\alpha}^*(\mc{E}_\sigma(\rho) \| \sigma)+2\log v(\sigma).
      \eeq
\end{enumerate}
\end{proposition}

\section{Problem formulation and main result}
\label{sec:main}

Let $(\rho_1,\sigma_1)$ and $(\rho_2,\sigma_2)$ be two quantum dichotomies. In the procedure of quantum dichotomies transformation, we use a quantum channel $\mathcal{E}$ to transform $(\rho_1,\sigma_1)$ into 
as much $(\rho_2,\sigma_2)$ as possible. We allow a small error in the transformation $\rho_1 \rightarrow \rho_2$, but the transformation $\sigma_1 \rightarrow \sigma_2$ is required to be exact. For $n \in \mathbb{N}$ and $\epsilon \in [0,1]$, we define the maximal achievable quantum dichotomies transformation number within error $\epsilon$ for $n$ copies of $(\rho_1,\sigma_1)$ as 
\begin{align}
&M^{\Delta, \epsilon}_{\rho_1,\sigma_1 \rightarrow \rho_2, \sigma_2}(n)\nonumber\\
&\quad:=\max\{m~:~ \exists~\text{a quantum channel}~\mathcal{E}~\text{such that}~\mathcal{E}(\sigma^{\ox n}_1)=\sigma^{\ox m}_2, \Delta(\mathcal{E}(\rho^{\ox n}_1),\rho^{\ox m}_2)\leq \epsilon\},
\end{align}
where $\Delta$ can be chosen as the trace distance $d$ or the purified distance $P$.

In the asymptotic setting in which multiple copies of $(\rho_1,\sigma_1)$ are available, the first order asymptotic behavior of $M^{\Delta, \epsilon}_{\rho_1,\sigma_1 \rightarrow \rho_2, \sigma_2}(n)$  has been derived in~\cite{FDT2019information, WangWilde2019resource} , i.e.,
\begin{equation}
\label{equ:defm}
\lim_{n \rightarrow \infty} \frac{1}{n} M^{\Delta, \epsilon}_{\rho_1,\sigma_1 \rightarrow \rho_2, \sigma_2}(n)=\frac{D(\rho_1\|\sigma_1)}{D(\rho_2\|\sigma_2)},
\end{equation}
where $\Delta$ can be either $d$ or $P$.

Eq.~(\ref{equ:defm}) implies that when transformation rate $r>\frac{D(\rho_1\|\sigma_1)}{D(\rho_2\|\sigma_2)}$, the optimal error among all transformation schemes 
\[
\epsilon^{\Delta,r}_{\rho_1,\sigma_1 \rightarrow \rho_2, \sigma_2}(n):=\min\{\Delta(\mathcal{E}_n(\rho^{\ox n}_1),\rho_2^{\ox nr})~:~\exists~\text{a quantum channel}~\mathcal{E}_n~\text{such that}~\mathcal{E}_n(\sigma^{\ox n}_1)=\sigma^{\ox nr}_2  \}
\]
converge to $1$ exponentially fast. The exact exponential decay rate is called as the strong converse exponent of quantum dichotomies and formally defined as follows.

\begin{definition}
For two quantum dichotomies $(\rho_1, \sigma_1)$, $(\rho_2, \sigma_2)$ and a transformation rate $r$, the strong converse exponents of quantum dichotomies based on the purified distance and the trace distance are defined respectively as
\begin{align*}
\begin{split}
E^P_{sc}(\rho_1, \sigma_1, \rho_2, \sigma_2, r):&=\lim_{n \rightarrow \infty} \frac{-1}{n} \log \big(1-\epsilon^{P,r}_{\rho_1,\sigma_1 \rightarrow \rho_2, \sigma_2}(n)\big),\\
E^d_{sc}(\rho_1, \sigma_1, \rho_2, \sigma_2, r):&=\lim_{n \rightarrow \infty} \frac{-1}{n} \log \big(1-\epsilon^{d,r}_{\rho_1,\sigma_1 \rightarrow \rho_2, \sigma_2}(n)\big).
\end{split}
\end{align*}
\end{definition}

\begin{remark}
\label{rem:equde}
By the relation between the purified distance and the fidelity function, $E^P_{sc}(\rho_1, \sigma_1, \rho_2, \sigma_2, r)$ can be equivalently expressed as 
\begin{equation}
E^P_{sc}(\rho_1, \sigma_1, \rho_2, \sigma_2, r)=\lim_{n \rightarrow \infty} \frac{-1}{n} \log \max_{\mathcal{E}_n: \mathcal{E}_n(\sigma_1^{\ox n})=\sigma_2^{\ox nr}} F^2(\mathcal{E}_n(\rho_1^{\ox n}), \rho_2^{\ox nr}).
\end{equation}
\end{remark}

Wang and Wilde have established a lower bound for $E^P_{sc}(\rho_1, \sigma_1, \rho_2, \sigma_2, r)$ in~\cite[Proposition~1]{WangWilde2019resource}, following their discussion in Eq.~(J33)--(J35) and (L11)--(J12). In particular, the following lower bound can be derived
\begin{equation}
\label{equ:mark}
E^P_{sc}(\rho_1, \sigma_1, \rho_2, \sigma_2, r) \geq \sup_{\frac{1}{2}\leq \alpha \leq 1} \frac{1-\alpha}{\alpha} \big\{rD_{\alpha}^*(\rho_2 \| \sigma_2)-D_{\frac{\alpha}{2\alpha-1}}^*(\rho_1\|\sigma_1)  \big\}.
\end{equation}

In this paper, we prove that the lower bound in Eq.~(\ref{equ:mark}) is also achievable, thus establishing the exact expression for $E^P_{sc}(\rho_1, \sigma_1, \rho_2, \sigma_2, r)$ and if $\rho_2$ is a pure state, we also determine the exact expression for $E^d_{sc}(\rho_1, \sigma_1, \rho_2, \sigma_2, r)$.  Our main results are as follows.

\begin{theorem}
For any $\rho_1, \sigma_1 \in \mc{S}(\mc{H}_1)$, $\rho_2, \sigma_2 \in \mc{S}(\mc{H}_2)$ that satisfy $\supp(\rho_1) \subseteq \supp(\sigma_1)$ or $\supp(\rho_2) \subseteq \supp(\sigma_2)$  and any $r>0$, we have
\begin{equation}
\label{equ:thm}
E^P_{sc}(\rho_1, \sigma_1, \rho_2, \sigma_2, r)=\sup_{\frac{1}{2}\leq \alpha \leq 1} \frac{1-\alpha}{\alpha} \big\{rD_{\alpha}^*(\rho_2 \| \sigma_2)-D_{\frac{\alpha}{2\alpha-1}}^*(\rho_1\|\sigma_1)  \big\}.
\end{equation}
In addition, if $\rho_2$ is a pure state, we have
\begin{equation}
\label{equ:strdis}
E^d_{sc}(\rho_1, \sigma_1, \rho_2, \sigma_2, r)=\sup_{\beta \geq 1} \frac{\beta-1}{\beta} \big\{rD_{\frac{1}{\beta}}(\rho_2 \| \sigma_2)-D^*_{\beta}(\rho_1 \| \sigma_1)   \big\}.
\end{equation}
\end{theorem}
\begin{remark}
By Lemma~\ref{lem:rela} in Appendix, if $\rho_2$ is a pure state, we have
\begin{equation}
\label{equ:rela}
\epsilon^{d,r}_{\rho_1,\sigma_1 \rightarrow \rho_2, \sigma_2}(n)\leq \epsilon^{P,r}_{\rho_1,\sigma_1 \rightarrow \rho_2, \sigma_2}(n)
\leq \sqrt{\epsilon^{d,r}_{\rho_1,\sigma_1 \rightarrow \rho_2, \sigma_2}(n)}.
\end{equation}
Eq.~(\ref{equ:rela}) implies that 
\begin{equation}
\label{equ:relputra}
1-\epsilon^{d,r}_{\rho_1,\sigma_1 \rightarrow \rho_2, \sigma_2}(n)
\geq 1-\epsilon^{P,r}_{\rho_1,\sigma_1 \rightarrow \rho_2, \sigma_2}(n)
\geq \frac{1-\epsilon^{d,r}_{\rho_1,\sigma_1 \rightarrow \rho_2, \sigma_2}(n)} {1+\sqrt{\epsilon^{d,r}_{\rho_1,\sigma_1 \rightarrow \rho_2, \sigma_2}(n)}}.
\end{equation}
If Eq.~(\ref{equ:thm}) holds, we can obtain from Eq.~(\ref{equ:relputra})
\begin{equation}
\begin{split}
&E^d_{sc}(\rho_1, \sigma_1, \rho_2, \sigma_2, r) \\
=&E^P_{sc}(\rho_1, \sigma_1, \rho_2, \sigma_2, r) \\
=&\sup_{\frac{1}{2}\leq  \alpha \leq 1} \frac{1-\alpha}{\alpha} \big\{rD_{\alpha}^{*}(\rho_2 \|\sigma_2)-D_{\frac{\alpha}{2\alpha-1}}^{*}(\rho_1 \| \sigma_1)\big\} \\
=&\sup_{\frac{1}{2}\leq  \alpha \leq 1} \frac{1-\alpha}{\alpha} \big\{rD_{2-\frac{1}{\alpha}}(\rho_2 \|\sigma_2)-D_{\frac{\alpha}{2\alpha-1}}^{*}(\rho_1 \| \sigma_1)\big\} \\
=&\sup_{\beta \geq 1} \frac{\beta-1}{\beta}\big\{rD_{\frac{1}{\beta}}(\rho_2 \|\sigma_2)-D_\beta^{*}(\rho_1 \| \sigma_1)\big\},
\end{split}
\end{equation}
where the fourth line is because that when $\rho_2$ is a pure state, $D_{\alpha}^*(\rho_2\|\sigma_2)=D_{2-\frac{1}{\alpha}}(\rho_2 \| \sigma_2)$. Hence, in the following, we only need to 
establish Eq.~(\ref{equ:thm}).
\end{remark}

\section{Proof of the achievability part}
\label{sec:achi}

In this section, we focus on the proof of the achievability part of Eq.~(\ref{equ:thm}). The process is divided into three steps. In the first step, we construct an intermediate quantity $F(\rho_1, \sigma_1, \rho_2, \sigma_2, r)$ and derive its 
variational expression. In the second step, we prove that $F(\rho_1, \sigma_1, \rho_2, \sigma_2, r)$ is an upper bound of $E^P_{sc}(\rho_1, \sigma_1, \rho_2, \sigma_2, r)$. In the last step, we improve this upper bound to the correct form.

\begin{definition}
Let $\rho_1, \sigma_1 \in \mc{S}(\mc{H}_1)$, $\rho_2, \sigma_2 \in \mc{S}(\mc{H}_2)$ and $r \in \mathbb{R}$. The quantity $F(\rho_1, \sigma_1, \rho_2, \sigma_2, r)$ is defined as
\begin{equation}
F(\rho_1, \sigma_1, \rho_2, \sigma_2, r):=\sup_{\frac{1}{2}<\alpha < 1} \frac{1-\alpha}{\alpha} \left\{rD_{\alpha}^{\flat}(\rho_2 \|\sigma_2)-D_\beta^{\flat}(\rho_1 \| \sigma_1)\right\},
\end{equation}
where $\frac{1}{\alpha}+\frac{1}{\beta}=2$.
\end{definition}

In the next proposition, we establish a variational expression for $F(\rho_1, \sigma_1, \rho_2, \sigma_2, r)$.
\begin{proposition}
\label{pro:var}
For $\rho_1, \sigma_1 \in \mc{S}(\mc{H}_1)$ that satisfy $\supp(\rho_1) \subset \supp(\sigma_1)$, $\rho_2, \sigma_2 \in \mc{S}(\mc{H}_2)$ and $r \geq 0$, we have
\[
F(\rho_1, \sigma_1, \rho_2, \sigma_2, r)
=\inf_{(\tau_1,\tau_2) \in \mc{S}_{\rho_1} \times \mc{S}_{\rho_2}} \left\{rD(\tau_2 \| \rho_2)+D(\tau_1\|\rho_1)+|rD(\tau_2 \| \sigma_2)-D(\tau_1 \| \sigma_1)|^+\right\}.
\]
\end{proposition}
\begin{proof}
By Proposition~\ref{prop:mainpro}(\romannumeral3), we have
\begin{equation}
\begin{split}
&\sup_{\frac{1}{2}< \alpha < 1} \frac{1-\alpha}{\alpha} \{rD_{\alpha}^{\flat}(\rho_2 \|\sigma_2)-D_\beta^{\flat}(\rho_1 \| \sigma_1) \} \\
=& \sup_{\frac{1}{2}< \alpha < 1} \frac{1-\alpha}{\alpha} \big\{r\inf_{\tau_2 \in \mc{S}_{\rho_2}} \big(D(\tau_2\|\sigma_2)-\frac{\alpha}{\alpha-1}D(\tau_2 \| \rho_2)\big)-\sup_{\tau_1 \in \mc{S}_{\rho_1}}\big(D(\tau_1\|\sigma_1)-\frac{\beta}{\beta-1}D(\tau_1 \| \rho_1) \big) \big\} \\
=& \sup_{\frac{1}{2}< \alpha < 1}\inf_{(\tau_1,\tau_2) \in \mc{S}_{\rho_1} \times \mc{S}_{\rho_2}} \frac{1-\alpha}{\alpha} \{rD(\tau_2\|\sigma_2)-\frac{\alpha}{\alpha-1}rD(\tau_2 \| \rho_2)-D(\tau_1\|\sigma_1)+\frac{\beta}{\beta-1}D(\tau_1 \| \rho_1)  \} \\
=&\sup_{\frac{1}{2}< \alpha < 1}\inf_{(\tau_1,\tau_2) \in \mc{S}_{\rho_1} \times \mc{S}_{\rho_2}} \big\{ rD(\tau_2\|\rho_2)+D(\tau_1 \| \rho_1)+ \frac{1-\alpha}{\alpha}\big(rD(\tau_2\|\sigma_2)-D(\tau_1\|\sigma_1) \big)   \big\} \\
=& \sup_{0<\delta<1}\inf_{(\tau_1,\tau_2) \in \mc{S}_{\rho_1} \times \mc{S}_{\rho_2}} \big\{ rD(\tau_2\|\rho_2)+D(\tau_1 \| \rho_1)+ \delta\big(rD(\tau_2\|\sigma_2)-D(\tau_1\|\sigma_1) \big)   \big\} \\
=&\inf_{(\tau_1,\tau_2) \in \mc{S}_{\rho_1} \times \mc{S}_{\rho_2}}\sup_{0<\delta<1} \big\{ rD(\tau_2\|\rho_2)+D(\tau_1 \| \rho_1)+ \delta\big(rD(\tau_2\|\sigma_2)-D(\tau_1\|\sigma_1) \big)   \big\} \\
=&\inf_{(\tau_1,\tau_2) \in \mc{S}_{\rho_1} \times \mc{S}_{\rho_2}} \big\{rD(\tau_2 \| \rho_2)+D(\tau_1\|\rho_1)+|rD(\tau_2 \| \sigma_2)-D(\tau_1 \| \sigma_1)|^+ \big\},
\end{split}
\end{equation}
where the sixth line follows from Sion's minimax theorem. The convexity of the function $(\tau_1, \tau_2) \mapsto rD(\tau_2\|\rho_2)+D(\tau_1 \| \rho_1)+ \delta\big(rD(\tau_2\|\sigma_2)-D(\tau_1\|\sigma_1)\big) $ that the Sion's minimax theorem requires is not obvious. To see this, we write
\begin{align}
\begin{split}
&rD(\tau_2\|\rho_2)+D(\tau_1 \| \rho_1)+ \delta\big(rD(\tau_2\|\sigma_2)-D(\tau_1\|\sigma_1)\big) \\
=&rD(\tau_2\|\rho_2)+\delta rD(\tau_2\|\sigma_2)-(1-\delta)H(\tau_1)+\tr \tau_1(\delta\log \sigma_1 -\log \rho_1).
\end{split}
\end{align}
Because $rD(\tau_2\|\rho_2)+\delta rD(\tau_2\|\sigma_2)$, $-(1-\delta)H(\tau_1)$ and $\tr \tau_1(\delta\log \sigma_1 -\log \rho_1)$ are all convex function of $(\tau_1, \tau_2)$, the result follows.
\end{proof}

Next, we show that $F(\rho_1, \sigma_1, \rho_2, \sigma_2, r)$ is an upper bound for $E^P_{sc}(\rho_1, \sigma_1, \rho_2, \sigma_2, r)$.
\begin{lemma}
\label{lem:main}
For $\rho_1, \sigma_1 \in \mc{S}(\mc{H}_1)$ that satisfy $\supp(\rho_1) \subset \supp(\sigma_1)$, $\rho_2, \sigma_2 \in \mc{S}(\mc{H}_2)$ and $r \geq 0$, we have
\[E^P_{sc}(\rho_1, \sigma_1, \rho_2, \sigma_2, r) \leq F(\rho_1, \sigma_1, \rho_2, \sigma_2, r).\]
\end{lemma}

In order to prove Lemma~\ref{lem:main}, we introduce
\begin{equation}
\begin{split}
&F_1(\rho_1, \sigma_1, \rho_2, \sigma_2, r):= \inf_{(\tau_1, \tau_2)\in \mc{F}_1}  \left\{rD(\tau_2\|\rho_2)+D(\tau_1 \|\rho_1)\right\} \\
&F_2(\rho_1, \sigma_1, \rho_2, \sigma_2, r):=\inf_{(\tau_1, \tau_2)\in \mc{F}_2} \left\{rD(\tau_2\|\rho_2)+D(\tau_1 \|\rho_1)+rD(\tau_2\|\sigma_2)-D(\tau_1\|\sigma_1)\right\},
\end{split}
\end{equation}
with
\begin{align*}
\begin{split}
&\mc{F}_1:=\left\{(\tau_1, \tau_2)~|~(\tau_1, \tau_2)\in \mc{S}_{\rho_1} \times \mc{S}_{\rho_2}, r<\frac{D(\tau_1 \|\sigma_1)}{D(\tau_2\|\sigma_2)} \right\},   \\
&\mc{F}_2:=\left\{(\tau_1, \tau_2)~|~(\tau_1, \tau_2)\in \mc{S}_{\rho_1} \times \mc{S}_{\rho_2}, r\geq \frac{D(\tau_1 \|\sigma_1)}{D(\tau_2\|\sigma_2)} \right\}.
\end{split}
\end{align*}
It is obviously seen from Proposition~\ref{pro:var} that
\begin{equation}
F(\rho_1, \sigma_1, \rho_2, \sigma_2, r)=\min\{F_1(\rho_1, \sigma_1, \rho_2, \sigma_2, r), F_2(\rho_1, \sigma_1, \rho_2, \sigma_2, r)  \}.
\end{equation}
So it is suffices to prove that $F_1(\rho_1, \sigma_1, \rho_2, \sigma_2, r)$ and $F_2(\rho_1, \sigma_1, \rho_2, \sigma_2, r)$ are both the upper bounds for $E^P_{sc}(\rho_1, \sigma_1, \rho_2, \sigma_2, r)$. We 
accomplish the proof in the following Lemma~\ref{lem:F1}  and Lemma~\ref{lem:F2} , respectively.

\begin{lemma}
\label{lem:F1}
For  $\rho_1, \sigma_1 \in \mc{S}(\mc{H}_1)$ that satisfy $\supp(\rho_1) \subset \supp(\sigma_1)$, $\rho_2, \sigma_2 \in \mc{S}(\mc{H}_2)$ and $r \geq 0$, we have
\[
E^P_{sc}(\rho_1, \sigma_1, \rho_2, \sigma_2, r) \leq F_1(\rho_1, \sigma_1, \rho_2, \sigma_2, r).
\]
\end{lemma}

\begin{proof}
According to the definition of $F_1(\rho_1, \sigma_1, \rho_2, \sigma_2, r)$, for any $\delta>0$, there exists $(\tau_1, \tau_2) \in \mc{S}_{\rho_1} \times \mc{S}_{\rho_2}$ such that
\begin{equation}
\label{F1:con}
r<\frac{D(\tau_1 \| \sigma_1)}{D(\tau_2 \| \sigma_2)}
\end{equation}
and 
\begin{equation}
rD(\tau_2 \| \rho_2)+D(\tau_1 \| \rho_1) \leq F_1(\rho_1, \sigma_1, \rho_2, \sigma_2, r) + \delta.
\end{equation}
Eq.~(\ref{F1:con}) implies that $r$ is an asymptotically achievable rate for the transformation from $(\tau_1, \sigma_1)$ to $(\tau_2, \sigma_2)$. Hence, there exists a sequence of quantum channels $\{\mc{E}_n\}_{n \in \mathbb{N}}$ such that
\begin{equation}
\begin{split}
&\mc{E}_n(\sigma_1^{\ox n})=\sigma_2^{\ox nr}, \\
\lim_{n \rightarrow \infty} &\|\mc{E}_n(\tau_1^{\ox n})-\tau_2^{\ox nr} \|_1=0.
\end{split}
\end{equation}

Now, we let $t=4(rD(\tau_2 \| \rho_2)+D(\tau_1 \| \rho_1))$, $\epsilon_n=2^{-nt}$, $C=\max\{1-r\log\lambda_{\rm{min}}(\rho_2), t+r\log |\mathcal{H}_2| \}$, where $\lambda_{\rm{min}}(\rho_2)$ is the smallest non-zero eigenvalue of $\rho_2$ and $\rho_2^n=(1-\epsilon_n)\rho_2^{\ox nr}+\epsilon_n \pi_n$, where 
$\pi_n$ is the maximally mixed state onto the subspace $\supp(\rho_2^{\ox nr})^{\perp}$. For sufficiently large $n$, $\lambda_{\rm{min}}(\rho_2^n) \geq 2^{-nC}$ and for such $n$, $F(\rho_2^n, \mc{E}_n(\rho_1^{\ox n}))$ can be evaluated as follows.
\begin{equation}
\label{F1:eva}
\begin{split}
&-\log F^2(\rho_2^n, \mc{E}_n(\rho_1^{\ox n})) \\
\leq &D(\mc{E}_n(\tau_1^{\ox n}) \| \mc{E}_n(\rho_1^{\ox n}))+D(\mc{E}_n(\tau_1^{\ox n}) \| \rho_2^n) \\
\leq &nD(\tau_1 \| \rho_1)+D(\mc{E}_n(\tau_1^{\ox n}) \| \rho_2^n) \\
=&nD(\tau_1 \| \rho_1)-H(\mc{E}_n(\tau_1^{\ox n}))-\tr\mc{E}_n(\tau_1^{\ox n})\log\rho_2^n+H(\tau_2^{\ox nr})+\tr\tau_2^{\ox nr}\log \rho_2^n+D(\tau_2^{\ox nr} \| \rho_2^n)\\
=&nD(\tau_1 \| \rho_1)-H(\mc{E}_n(\tau_1^{\ox n}))+H(\tau_2^{\ox nr})+\tr\big(\tau_2^{\ox nr}-\mc{E}_n(\tau_1^{\ox n})\big)\log\rho_2^n+D(\tau_2^{\ox nr} \| \rho_2^n),
\end{split}
\end{equation}
where the second line comes from Lemma~\ref{lem:fidelity-re} in Appendix, the third line is due to the data processing inequality. By making use of Fannes-Audenaert inequality and operator H\"{o}lder inequality, we can further upper bound Eq.~(\ref{F1:eva}) as 
\begin{equation}
\label{F1:eva11}
\begin{split}
&-\log F^2(\rho_2^n, \mc{E}_n(\rho_1^{\ox n})) \\
\leq &nD(\tau_1 \| \rho_1)+\frac{1}{2}\|\mc{E}_n(\tau_1^{\ox n})-\tau_2^{\ox nr} \|_1\log |\mathcal{H}_2|^{nr}+h(\frac{1}{2}\|\mc{E}_n(\tau_1^{\ox n})-\tau_2^{\ox nr} \|_1)\\
&+\|\mc{E}_n(\tau_1^{\ox n})-\tau_2^{\ox nr} \|_1 \|\log \rho_2^n\|_\infty+D(\tau_2^{\ox nr} \| \rho_2^n) \\
\leq &nD(\tau_1 \| \rho_1)+\frac{1}{2}\|\mc{E}_n(\tau_1^{\ox n})-\tau_2^{\ox nr} \|_1\log |\mathcal{H}_2|^{nr}+h(\frac{1}{2}\|\mc{E}_n(\tau_1^{\ox n})-\tau_2^{\ox nr} \|_1)\\
&+nC\|\mc{E}_n(\tau_1^{\ox n})-\tau_2^{\ox nr} \|_1+nrD(\tau_2 \| \rho_2)-\log(1-\epsilon_n).
\end{split}
\end{equation}
Dividing $n$ in both sides of Eq.~(\ref{F1:eva11}) and taking the limit of $n$, we have
\begin{equation}
\label{evafia}
\lim_{n \rightarrow \infty} -\frac{\log F^2(\rho_2^n, \mc{E}_n(\rho_1^{\ox n}))}{n} \leq D(\tau_1 \| \rho_1)+rD(\tau_2 \| \rho_2).
\end{equation}

By Lemma~\ref{lem:rea} in Appendix, the following inequality holds:
\begin{equation}
\label{equ:fin}
\begin{split}
F^2( \mc{E}_n(\rho_1^{\ox n}), \rho_2^{\ox nr})&\geq  F^2(\rho_2^n, \mc{E}_n(\rho_1^{\ox n}))-\sqrt{\|\rho_2^n-\rho_2^{\ox nr}\|_1 }\\
&=  F^2(\rho_2^n, \mc{E}_n(\rho_1^{\ox n}))-\sqrt{2\epsilon_n}.
\end{split}
\end{equation}
Eq.~(\ref{evafia}), Eq.~(\ref{equ:fin}) and the definition of $\epsilon_n$ give that
\begin{equation}
\label{equ:intla}
\lim_{n \rightarrow \infty} -\frac{\log F^2( \mc{E}_n(\rho_1^{\ox n}), \rho_2^{\ox nr})}{n} \leq D(\tau_1 \| \rho_1)+rD(\tau_2 \| \rho_2) \leq F_1(\rho_1, \sigma_1, \rho_2, \sigma_2, r)+\delta.
\end{equation}
Eq.~(\ref{equ:intla}) and Remark~\ref{rem:equde} imply that
\begin{equation}
\label{equ:F1la}
E^P_{sc}(\rho_1, \sigma_1, \rho_2, \sigma_2, r) \leq F_1(\rho_1, \sigma_1, \rho_2, \sigma_2, r) +\delta.
\end{equation}
Because Eq.~(\ref{equ:F1la}) holds for any $\delta>0$, by letting $\delta \rightarrow 0$, we complete the proof.
\end{proof}

\begin{lemma}
\label{lem:F2}
For  $\rho_1, \sigma_1 \in \mc{S}(\mc{H}_1)$ that satisfy $\supp(\rho_1) \subset \supp(\sigma_1)$, $\rho_2, \sigma_2 \in \mc{S}(\mc{H}_2)$ and $r \geq 0$, we have
\[
E^P_{sc}(\rho_1, \sigma_1, \rho_2, \sigma_2, r) \leq F_2(\rho_1, \sigma_1, \rho_2, \sigma_2, r).
\]
\end{lemma}

\begin{proof}
According to the definition of $F_2(\rho_1, \sigma_1, \rho_2, \sigma_2, r)$, there exists $(\tau_1,\tau_2) \in \mc{S}_{\rho_1} \times \mc{S}_{\rho_2}$  such that
\begin{equation}
r \geq \frac{D(\tau_1 \| \sigma_1)}{D(\tau_2 \| \sigma_2)}
\end{equation}
and
\begin{equation}
\label{equ:f2}
F_2(\rho_1, \sigma_1, \rho_2, \sigma_2, r)=rD(\tau_2\| \rho_2)+D(\tau_1 \| \rho_1)+rD(\tau_2 \| \sigma_2)-D(\tau_1 \| \sigma_1).
\end{equation}

We let $r_1=\frac{D(\tau_1 \| \sigma_1)}{D(\tau_2 \| \sigma_2)}-\delta$, where $\delta>0$ is a constant. From the proof of Lemma~\ref{lem:F1}, we know that there exists a sequence of quantum channels $\{\mc{E}_n \}_{n \in \mathbb{N}}$ such that
\begin{equation}
\mc{E}_n(\sigma_1^{\ox n})=\sigma_2^{\ox nr_1}
\end{equation}
and
\begin{equation}
\lim_{n \rightarrow \infty} -\frac{\log F^2(\rho_2^{\ox nr_1}, \mc{E}_n(\rho_1^{\ox n}))}{n} \leq D(\tau_1 \| \rho_1)+r_1D(\tau_2 \| \rho_2).
\end{equation}

Now, we construct a sequence of quantum channels $\{\Phi_n\}_{n \in \mathbb{N}}$ as 
\[
\Phi_n(X)=\mc{E}_n(X) \ox \sigma_2^{\ox n(r-r_1)}
\]
It is easy to verify that
\begin{equation}
\label{equ:F21}
\Phi_n(\sigma_1^{\ox n})=\sigma_2^{\ox nr}
\end{equation}
and 
\begin{equation}
\label{equ:F22}
\begin{split}
&\lim_{n \rightarrow \infty} -\frac{\log F^2(\rho_2^{\ox nr}, \Phi_n(\rho_1^{\ox n}))}{n} \\
= &\lim_{n \rightarrow \infty}-\bigg\{ \frac{\log F^2(\rho_2^{ \ox nr_1},\mc{E}_n(\rho_1^{\ox n}))}{n}+ \frac{\log F^2(\rho_2^{\ox n(r-r_1)},\sigma_2^{\ox n(r-r_1)})}{n}   \bigg\} \\
\leq &D(\tau_1 \| \rho_1)+r_1D(\tau_2 \| \rho_2)-(r-r_1)\log F^2(\rho_2, \sigma_2) \\
\leq &D(\tau_1 \| \rho_1)+r_1D(\tau_2 \| \rho_2)+(r-r_1)\left\{D(\tau_2\| \rho_2)+D(\tau_2 \| \sigma_2)\right\}  \\
= &D(\tau_1 \| \rho_1)+rD(\tau_2 \| \rho_2)+rD(\tau_2 \| \sigma_2)-D(\tau_1\|\sigma_1)+\delta D(\tau_2 \| \sigma_2),
\end{split}
\end{equation}
where in second inequality, we use Lemma~\ref{lem:fidelity-re} in Appendix. Eq.~(\ref{equ:f2}), Eq.~(\ref{equ:F21}) and Eq.~(\ref{equ:F22}) imply that 
\begin{equation}
\label{equ:f2la}
E^P_{sc}(\rho_1, \sigma_1, \rho_2, \sigma_2, r) \leq F_2(\rho_1, \sigma_1, \rho_2, \sigma_2, r) +\delta D(\tau_2 \| \sigma_2).
\end{equation}
Noticing that Eq.~(\ref{equ:f2la}) holds for any $\delta>0$, by letting $\delta \rightarrow 0$,  we complete the proof.
\end{proof}

In the last step, we improve the the upper bound in Lemma~\ref{lem:main} to the correct form to finish the proof of the achievability part of Eq.~(\ref{equ:thm}).
\begin{proofof}{the achievability part of Eq.~(\ref{equ:thm})}
If $\supp(\rho_1) \subseteq \supp(\sigma_1)$ does not hold, it is obvious that all terms in Eq.~(\ref{equ:thm}) equal to $0$ and the assertion holds trivially. Hence, for the rest we assume that $\supp(\rho_1) \subseteq \supp(\sigma_1)$.
We first fix an integer $m$ and let $\rho_1^m=\mc{E}_{\sigma_1^{\ox m}}(\rho_1^{\ox m})$, $\rho_2^m=\mc{E}_{\sigma_2^{\ox m}}(\rho_2^{\ox m})$. From Lemma~\ref{lem:main} , we know that there exists a sequence of quantum channels $\{\Phi_k\}_{k \in \mathbb{N}}$ such that
\begin{equation}
\label{equ:2}
\Phi_k(\sigma_1^{\ox mk})=\sigma_2^{\ox mkr}
\end{equation}
and 
\begin{equation}
\label{equ:imin}
\lim_{k \rightarrow \infty} -\frac{\log F^2({\rho^m_2}^{\ox kr}, \Phi_k({\rho^m_1}^{\ox k}))}{k} \leq \sup_{\frac{1}{2}< \alpha < 1} \frac{1-\alpha}{\alpha} \big\{rD_{\alpha}^{\flat}(\rho^m_2 \|\sigma^{\ox m}_2)-D_\beta^{\flat}(\rho^m_1 \| \sigma^{\ox m}_1)\big\}.
\end{equation}
Lemma~\ref{lem:appen1} in Appendix and Eq.~(\ref{equ:imin}) imply that 
\begin{equation}
\label{equ:las}
\begin{split}
&\lim_{k \rightarrow \infty} -\frac{\log F^2({\rho_2}^{\ox mkr}, \mc{E}_{\sigma_2^{\ox m}}^{\ox kr} \circ \Phi_k({\rho^m_1}^{\ox k}))}{k} \\
\leq &\lim_{k \rightarrow \infty} -\frac{\log  v^{-kr}(\sigma_2^{\ox m})F^2({\rho^m_2}^{\ox kr}, \mc{E}_{\sigma_2^{\ox m}}^{\ox kr} \circ \Phi_k({\rho^m_1}^{\ox k}))}{k} \\
\leq &\lim_{k \rightarrow \infty} -\frac{\log  v^{-kr}(\sigma_2^{\ox m})F^2({\rho^m_2}^{\ox kr},\Phi_k({\rho^m_1}^{\ox k}))}{k} \\
\leq &\sup_{\frac{1}{2}< \alpha < 1} \frac{1-\alpha}{\alpha} \big\{rD_{\alpha}^{\flat}(\rho^m_2 \|\sigma^{\ox m}_2)-D_\beta^{\flat}(\rho^m_1 \| \sigma^{\ox m}_1)\big\}+r \log v(\sigma_2^{\ox m}).
\end{split}
\end{equation}

Now, for any integer $n$, we can write $n=km+l$, where $0 \leq l \leq m-1$ and we construct a quantum channel $\mc{E}_n$ as 
\[
\mc{E}_n(X)=\big(\mc{E}_{\sigma_2^{\ox m}}^{\ox kr} \circ \Phi_k  \circ \mc{E}_{\sigma_1^{\ox m}}^{\ox k} \circ \tr_l(X)\big)\ox \sigma_2^{\ox lr},
\]
where $\tr_l$ denotes taking the partial trace over the last $l$ systems.

From Eq.~(\ref{equ:2}) and Eq.~(\ref{equ:las}), it is easy to see that 
\begin{equation}
\label{equ:pin1}
\mc{E}_n(\sigma_1^{\ox n})=\sigma_2^{\ox nr}
\end{equation}
and 
\begin{equation}
\label{equ:pin2}
\begin{split}
&\lim_{n \rightarrow \infty} -\frac{\log F^2({\rho_2}^{\ox nr}, \mc{E}_n(\rho_1^{\ox n}))}{n}  \\
=&\lim_{k \rightarrow \infty} -\frac{\log F^2({\rho_2}^{\ox mkr},\mc{E}_{\sigma_2^{\ox m}}^{\ox kr} \circ \Phi_k({\rho^m_1}^{\ox k}))F^2(\rho_2^{\ox lr},\sigma_2^{\ox lr})}{mk+l}\\
\leq &\sup_{\frac{1}{2}< \alpha < 1} \frac{1-\alpha}{\alpha} \bigg\{r\frac{D_{\alpha}^{\flat}(\rho^m_2 \|\sigma^{\ox m}_2)}{m}-\frac{D_\beta^{\flat}(\rho^m_1 \| \sigma^{\ox m}_1)}{m}\bigg\}+r \frac{\log v(\sigma_2^{\ox m})}{m}.
\end{split}
\end{equation}
Eq.~(\ref{equ:pin1}) and Eq.~(\ref{equ:pin2}) give that
\begin{equation}
\label{equ:imlast}
\begin{split}
&E^P_{sc}(\rho_1, \sigma_1, \rho_2, \sigma_2, r) \\
\leq &\sup_{\frac{1}{2}< \alpha < 1} \frac{1-\alpha}{\alpha} \left\{r\frac{D_{\alpha}^{\flat}(\rho^m_2 \|\sigma^{\ox m}_2)}{m}-\frac{D_\beta^{\flat}(\rho^m_1 \| \sigma^{\ox m}_1)}{m}\right\}+r \frac{\log v(\sigma_2^{\ox m})}{m} \\
=&\sup_{\frac{1}{2}< \alpha < 1} \frac{1-\alpha}{\alpha} \left\{r\frac{D_{\alpha}^{*}(\rho^m_2 \|\sigma^{\ox m}_2)}{m}-\frac{D_\beta^{*}(\rho^m_1 \| \sigma^{\ox m}_1)}{m}\right\}+r \frac{\log v(\sigma_2^{\ox m})}{m}\\
\leq &\sup_{\frac{1}{2}< \alpha < 1} \frac{1-\alpha}{\alpha} \left\{rD_{\alpha}^{*}(\rho_2 \|\sigma_2)-D_\beta^{*}(\rho_1 \| \sigma_1)+2\frac{v(\sigma_1^{\ox m})}{m}\right\}+r \frac{\log v(\sigma_2^{\ox m})}{m}\\
\leq &\sup_{\frac{1}{2}\leq \alpha \leq 1} \frac{1-\alpha}{\alpha} \left\{rD_{\alpha}^{*}(\rho_2 \|\sigma_2)-D_\beta^{*}(\rho_1 \| \sigma_1)\right\}+2\frac{v(\sigma_1^{\ox m})}{m}+r \frac{\log v(\sigma_2^{\ox m})}{m} \\
= &\sup_{\frac{1}{2}\leq \alpha \leq 1} \frac{1-\alpha}{\alpha} \big\{rD_{\alpha}^{*}(\rho_2 \|\sigma_2)-D_{\frac{\alpha}{2\alpha-1}}^{*}(\rho_1 \| \sigma_1)\big\}+2\frac{v(\sigma_1^{\ox m})}{m}+r \frac{\log v(\sigma_2^{\ox m})}{m}, 
\end{split}
\end{equation}
where the third line is because that $\rho^m_1$ and $\rho^m_2$ commute with $\sigma_1^{\ox m}$ and $\sigma_2^{\ox m}$, respectively, the fourth line is due to the data processing inequality and Proposition~\ref{prop:mainpro} (\romannumeral5). Because Eq.~(\ref{equ:imlast}) holds 
for any $m$, by letting $m \rightarrow \infty$, we complete the proof.
\end{proofof}

\section{Discussion}
\label{sec:conclu}

In this work, we determine the exact strong converse exponent of quantum dichotomies based on the purified distance. The combination of this result and an improved Fuchs–van de Graaf inequality also directly leads to the strong converse exponent based on the trace distance when $\rho_2$ is a pure state. We further believe that this formula still holds when $\rho_2$ is a general quantum state, but leave it as an open question.  While recent results have established lower bounds for the strong converse exponents under the trace distance for some tasks~\cite{Wilde2022distinguishability, SalzmannDatta2022total, SGC2022optimal}, deriving a matching upper bound seems a difficult problem. The improved Fuchs–van de Graaf inequality provides a new idea to solve this kind of problem when dealing with pure states. In this case, the strong converse exponent under the trace distance can be transformed into that under the purified distance. For determining the strong converse exponent under the purified distance, there are then already a few techniques developed~\cite{MosonyiOgawa2017strong, LiYao2024operational}.

\section*{Acknowledgments}

The authors would like to thank Mark Wilde for comments. MB and YY acknowledge funding by the European Research Council (ERC Grant Agreement No. 948139) and MB acknowledges support from the Excellence Cluster - Matter and Light for Quantum Computing (ML4Q).

\section{Appendix}

The following lemma is an improved  Fuchs–van de Graaf inequality.

\begin{lemma}
\label{lem:rela}
For any pure state $\varphi$ and quantum state $\rho$, we have
\begin{equation}
P(\rho, \varphi) \leq \sqrt{d(\rho, \varphi)}.
\end{equation}
\end{lemma}

\begin{proof}
 By direct calculation, we have
\begin{equation}
d(\rho, \varphi)=\tr(\varphi-\rho)_+ \geq  \tr(\varphi-\rho)\varphi=1-\tr\rho \varphi=1-F^2(\rho, \varphi)=P^2(\rho, \varphi).
\end{equation}
We are done.
\end{proof}

\begin{lemma}
\label{lem:rea}
For $\rho, \sigma, \tau \in \mc{S}(\mc{H})$, we have
\begin{equation}
F^2(\tau,\sigma) \geq F^2(\rho, \sigma) -\sqrt{\|\rho-\tau\|_1}.
\end{equation}
\end{lemma}

\begin{proof}
By Ulhmann's theorem~\cite{Uhlmann1976transition}, there exists purified states $\phi_\rho$, $\phi_\sigma$ and $\phi_\tau$ of $\rho$, $\sigma$ and $\tau$ such that
\[
F(\rho,\sigma)=F(\phi_\rho, \phi_\sigma), F(\rho, \tau)=F(\phi_\rho, \phi_\tau).
\]
It is easy to see that 
\begin{equation}
\begin{split}
F^2(\tau, \sigma) &\geq F^2(\phi_\tau, \phi_\sigma) \\
&=\tr \phi_\tau \phi_\sigma \\
& \geq \tr  \phi_\rho\phi_\sigma-\tr(\phi_\rho- \phi_\tau)_+ \\
&=F^2(\rho, \sigma)-P(\phi_\rho,\phi_\tau)\\
&=F^2(\rho, \sigma)-P(\rho, \tau)\\
&\geq F^2(\rho, \sigma) -\sqrt{\|\rho-\tau\|_1}.
\end{split}
\end{equation}
\end{proof}

The following three lemmas are used in our proofs and they have been proved in~\cite[Lemma~30]{LiYao2024operational} and \cite[Lemma~31]{LiYao2024operational}, respectively.

\begin{lemma}
\label{lem:appen1}
Let $\rho, \sigma\in\mc{S}(\mc{H})$, and let $\mc{H}=\bigoplus_{i\in \mc{I}}\mc{H}_i$ decompose into a set of mutually orthogonal subspaces $\{\mc{H}_i\}_{i\in \mc{I}}$. Suppose that $\sigma=\sum\limits_{i \in \mc{I}} \sigma_i$ with $\supp(\sigma_i)\subseteq \mc{H}_i$. Then
\begin{equation}
F\left(\sum_{i \in \mc{I}} \Pi_i \rho \Pi_i, \sigma\right) \leq \sqrt{|\mc{I}|} F(\rho, \sigma),
\end{equation}
where $\Pi_i$ is the projection onto $\mc{H}_i$.
\end{lemma}

\begin{lemma}
\label{lem:fidelity-re}
Let $\rho, \sigma, \tau \in\mc{S}(\mc{H})$ be any quantum states. Then we have
\begin{equation}\label{eq:fid-re}
-\log F^2(\rho,\sigma) \leq D(\tau\|\rho) + D(\tau\|\sigma).
\end{equation}
\end{lemma}


\begin{thebibliography}{99}

\bibitem{FDT2019information}
Francesco, B., David, S., Tomamichel, M.: An information-theoretic treatment of quantum dichotomies. Quantum. {\bf 3}, 
209--216 (2019)

\bibitem{WangWilde2019resource}
Wang, X., Wilde, M.M.: Resource theory of asymmetric distinguishability. Phys.
  Rev. Research {\bf 1}, 033170 (2019)

\bibitem{LPCRTK2024quantum}
Bartosic, P., Chubb, C., Renes, J., Tomamichel, M., Korzekwa, K.: Quantum dichotomies and coherent thermodynamics beyond first-order asymptotics. PRX Quantum. 
{\bf 5}(2), 020335 (2024)

\bibitem{KCT2019avoid}
Korzekwa, K., Chubb, C., Tomamichel, M.: Avoiding irreversibility: Engineering resonant conversions of quantum resources. Phys. Rev. Lett. 
{\bf 122}(11), 110403 (2019)

\bibitem{KCT2018Beyond}
Korzekwa, K., Chubb, C., Tomamichel, M.: Beyond the thermodynamic limit:
finite-size corrections to state interconversion rates. Quantum. 
{\bf 2}, 108--139 (2018)

\bibitem{KCT2019Moderate}
Korzekwa, K., Chubb, C., Tomamichel, M.: Moderate deviation analysis
of majorization-based resource interconversion. Phys. Rev. A. 
{\bf 99}(3), 032332 (2019)

\bibitem{KumagaiHayashi2013second}
Kumagai, W., Hayashi, M.: Second order asymptotics of optimal approximate conversion for probability distributions and entangled states and its application to LOCC cloning. arxiv:1306.4166 [quant-ph]. 

\bibitem{Wilde2022distinguishability}
Wilde, M.M.: On distinguishability distillation and dilution exponents. Quantum
  Inf. Process. {\bf 21}, 392 (2022)

\bibitem{AKCMBMA2007discriminating}
Audenaert, K.M., Calsamiglia, J., Munoz-Tapia, R., Bagan, E., Masanes, L., Acin, A.,
  Verstraete, F.: Discriminating states: the quantum {Chernoff} bound. Phys.
  Rev. Lett. {\bf 98}, 160501 (2007)

\bibitem{NussbaumSzkola2009chernoff}
Nussbaum, M., Szkola, A.: The {Chernoff} lower bound for symmetric quantum
  hypothesis testing. Ann. Statist. {\bf 37}(2), 1040--1057 (2009)

\bibitem{Li2016discriminating}
Li, K.: Discriminating quantum states: the multiple {Chernoff} distance. Ann.
  Statist. {\bf 44}(4), 1661--1679 (2016)

\bibitem{Nagaoka2006converse}
Nagaoka, H.: The converse part of the theorem for quantum {Hoeffding} bound.
  arXiv:quant-ph/0611289 (2006)

\bibitem{Hayashi2007error}
Hayashi, M.: Error exponent in asymmetric quantum hypothesis testing and its
  application to classical-quantum channel coding. Phys. Rev. A {\bf 76},
  062301 (2007)

\bibitem{ANSV2008asymptotic}
Audenaert, K.M., Nussbaum, M., Szkola, A., Verstraete, F.: Asymptotic error
  rates in quantum hypothesis testing. Commun. Math. Phys. {\bf 279}(1),
  251--283 (2008)


\bibitem{Renes2022achievable}
Renes, J.M.: Achievable error exponents of data compression with quantum side
  information and communication over symmetric classical-quantum channels.
  arXiv:2207.08899 [quant-ph]

\bibitem{CHDH2020non}
Cheng, H.-C., Hanson, E.P., Datta, N., Hsieh, M.-H.: Non-asymptotic classical data
  compression with quantum side information. IEEE Trans. Inf. Theory {\bf 67}(2),
  902--930 (2020)

\bibitem{HayashiTomamichel2016correlation}
Hayashi, M., Tomamichel, M.: Correlation detection and an operational
  interpretation of the {R{\'e}nyi} mutual information. J. Math. Phys. {\bf 57}(10),
  102201 (2016)

\bibitem{MosonyiOgawa2017strong}
Mosonyi, M., Ogawa, T.: Strong converse exponent for classical-quantum channel
  coding. Commun. Math. Phys. {\bf 355}(1), 373--426 (2017)

\bibitem{LiYao2024operational}
Li, K., Yao, Y.: Operational interpretation of the sandwiched R{\'e}nyi divergence of order 1/2 to 1 as strong converse exponents. Commun. Math. Phys.
 {\bf 405}(2), 22--67 (2024)

\bibitem{Hayashi2002optimal}
Hayashi, M.: Optimal sequence of quantum measurements in the sense of {Stein's}
  lemma in quantum hypothesis testing. J. Phys. A: Math. Gen. {\bf 35}, 10759 (2002)

\bibitem{Renyi1961measures}
R{\'e}nyi, A.: On measures of entropy and information. In: Proc. 4th Berkeley
  Symp. Math. Statist. Probab., vol.~1, pp. 547--561. University of California Press,
  Berkeley (1961)

\bibitem{MDSFT2013on}
M{\"u}ller-Lennert, M., Dupuis, F., Szehr, O., Fehr, S., Tomamichel, M.: On quantum
  {R{\'e}nyi} entropies: a new generalization and some properties. J. Math.
  Phys. {\bf 54}(12), 122203 (2013)

\bibitem{WWY2014strong}
Wilde, M.M., Winter, A., Yang, D.: Strong converse for the classical capacity of
  entanglement-breaking and {Hadamard} channels via a sandwiched {R{\'e}nyi}
  relative entropy. Commun. Math. Phys. {\bf 331}(2), 593--622 (2014)

\bibitem{Petz1986quasi}
Petz, D.: Quasi-entropies for finite quantum systems. Rep. Math. Phys. {\bf 23}(1),
  57--65 (1986)

\bibitem{HiaiPetz1993the}
Hiai, F., Petz, D.: The {Golden-Thompson} trace inequality is complemented.
  Linear Algbra Appl. {\bf 181}, 153--185 (1993)

\bibitem{MosonyiOgawa2015quantum}
Mosonyi, M., Ogawa, T.: Quantum hypothesis testing and the operational
  interpretation of the quantum {R{\'e}nyi} relative entropies. Commun. Math.
  Phys. {\bf 334}(3), 1617--1648 (2015)


\bibitem{GuptaWilde2015multiplicativity}
Gupta, M.K., Wilde, M.M.: Multiplicativity of completely bounded $p$-norms
  implies a strong converse for entanglement-assisted capacity. Commun. Math.
  Phys. {\bf 334}(2), 867--887 (2015)

\bibitem{CMW2016strong}
Cooney, T., Mosonyi, M., Wilde, M.M.: Strong converse exponents for a quantum
  channel discrimination problem and quantum-feedback-assisted communication.
  Commun. Math. Phys. {\bf 344}(3), 797--829 (2016)

\bibitem{LiYao2024strong}
Li, K., Yao, Y.: Strong converse exponent for entanglement-assisted
  communication. IEEE Tran. Inf. Theory. {\bf 70}(7), 5017--5029 (2024)


\bibitem{LYH2023tight}
Li, K., Yao, Y., Hayashi, M.: Tight exponential analysis for smoothing the
  max-relative entropy and for quantum privacy amplification. IEEE Trans. Inf.
  Theory {\bf 69}(3), 1680--1694 (2023)

\bibitem{LiYao2024reliability}
Li, K., Yao, Y.: Reliability function of quantum information decoupling via the
  sandwiched {R\'enyi} divergence. Commun. Math.
  Phys. {\bf 405}, 160--189 (2024)

\bibitem{Hayashi2015precise}
Hayashi, M.: Precise evaluation of leaked information with secure randomness
  extraction in the presence of quantum attacker. Commun. Math. Phys. {\bf 333}(1),
  335--350 (2015)

\bibitem{LiYao2021reliable}
Li, K., Yao, Y.: Reliable simulation of quantum channels. arXiv:2112.04475 [quant-ph]


\bibitem{Umegaki1954conditional}
Umegaki, H.: Conditional expectation in an operator algebra. Tohoku Math. J.
  {\bf 6}(2), 177--181 (1954)

\bibitem{Beigi2013sandwiched}
Beigi, S.: Sandwiched {R{\'e}nyi} divergence satisfies data processing
  inequality. J. Math. Phys. {\bf 54}(12), 122202 (2013)


\bibitem{FrankLieb2013monotonicity}
Frank, R.L., Lieb, E.H.: Monotonicity of a relative {R{\'e}nyi} entropy. J.
  Math. Phys. {\bf 54}(12), 122201 (2013)


\bibitem{SalzmannDatta2022total}
Robert, S., Nilanjana, D.: Total insecurity of communication via strong converse
  for quantum privacy amplification. arXiv:2202.11090 [quant-ph]


\bibitem{SGC2022optimal}
Shen, Y.-C., Gao, L., Cheng, H.-C.: Optimal second-order rates for quantum soft
  covering and privacy amplification. arXiv:2202.11590 [quant-ph]

\bibitem{Uhlmann1976transition}
Uhlmann, A.: The ``transition probability'' in the state space of a
  $^\ast$-algebra. Rep. Math. Phys. {\bf 9}(2), 273--279 (1976)

\end{thebibliography}
\end{document}